%% file: main.tex
\documentclass[a4paper,UKenglish,cleveref, autoref, thm-restate]{lipics-v2021}
\pdfoutput=1 %uncomment to ensure pdflatex processing (mandatatory e.g. to submit to arXiv)
\hideLIPIcs  %uncomment to remove references to LIPIcs series (logo, DOI, ...), e.g. when preparing a pre-final version to be uploaded to arXiv or another public repository

\input{head}

\title{Reverse-Robust Computation with Chemical Reaction Networks} %TODO Please add

\titlerunning{Reverse-Robust Computation with Chemical Reaction Networks} %TODO optional, please use if title is longer than one line

\author{Ravi Kini}{Computer Science, University of California, Davis}{kkini@ucdavis.edu}{https://orcid.org/0009-0009-7229-1795}{}
\author{David Doty}{Computer Science, University of California, Davis}{doty@ucdavis.edu}{https://orcid.org/0000-0002-3922-172X}{funded by NSF awards 2211793 and 2329909}%TODO mandatory, please use full name; only 1 author per \author macro; first two parameters are mandatory, other parameters can be empty. Please provide at least the name of the affiliation and the country. The full address is optional. Use additional curly braces to indicate the correct name splitting when the last name consists of multiple name parts.

\authorrunning{D. Doty and R. Kini}

\Copyright{David Doty and Ravi Kini}

\ccsdesc[500]{Theory of computation~Models of computation}

\keywords{chemical reaction networks, reverse-robust computation, semilinear}

\category{} %optional, e.g. invited paper

\relatedversion{} %optional, e.g. full version hosted on arXiv, HAL, or other respository/website
\nolinenumbers %uncomment to disable line numbering

\EventEditors{John Q. Open and Joan R. Access}
\EventNoEds{2}
\EventLongTitle{42nd Conference on Very Important Topics (CVIT 2016)}
\EventShortTitle{CVIT 2016}
\EventAcronym{CVIT}
\EventYear{2016}
\EventDate{December 24--27, 2016}
\EventLocation{Little Whinging, United Kingdom}
\EventLogo{}
\SeriesVolume{42}
\ArticleNo{23}
\begin{document}

\maketitle

\begin{abstract}
Chemical reaction networks, or CRNs, are known to stably compute semilinear Boolean-valued predicates and functions, provided that all reactions are irreversible. However, this property does not hold for wet-lab implementations, as all chemical reactions are reversible, even at very slow rates. We study the computational power of CRNs under the \emph{reverse-robust} computation model, where reactions are permitted to occur either in forward or in reverse up to a cutoff point, after which they may only occur in forward. Our main results show that all semilinear predicates and all semilinear functions can be computed reverse-robustly, and in fact, that existing constructions continue to hold under the reverse-robust computational model. A key tool used to prove correctness under the reverse-robust computation model is \emph{invariants}: linear (or linear modulo some $m$) combinations of the counts of the species that are preserved by all reactions.
\end{abstract}

\input{intro}
\input{prelim}
\input{results}
\input{conclusion}

\bibliography{ref}

\end{document}

%% file: head.tex
\usepackage{amsfonts,amsmath,amssymb}
\usepackage[
    textsize=scriptsize,
]{todonotes} 
\setuptodonotes{fancyline}
\def\rxn{\mathop{\rightarrow}\limits}
\renewcommand{\vec}[1]{\mathbf{#1}}

\newcommand{\vc}{\vec{c}}
\newcommand{\vy}{\vec{y}}

\newcommand{\reach}{\rightarrow}
\newcommand{\bireach}{\looparrowright}

\newcommand{\rev}[1]{{#1}^{\leftarrow}}

\newcommand{\reachrxn}[1]{\rightarrow^{#1}}

\def\longrightharpoonup{\relbar\joinrel\rightharpoonup}
\def\longleftharpoondown{\leftharpoondown\joinrel\relbar}
\def\longrightleftharpoons{\mathop{\vcenter{\hbox{\ooalign{\raise1pt\hbox{$\longrightharpoonup\joinrel$}\crcr\lower1pt\hbox{$\longleftharpoondown\joinrel$}}}}}}
\def\rxn{\mathop{\rightarrow}\limits}  % use as A+B \rxn^k C

\def\revrxn{\mathop{\rightleftharpoons}\limits}

%% file: intro.tex
\section{Introduction}
% \todo{DD: ensure we use ``state'' everywhere instead of ``configuration''}
Chemical reaction networks (CRNs) are a fundamental tool for understanding and designing molecular systems. By abstracting chemical reactions into a set of finite, rule-based transformations, CRNs allow us to model the behavior of complex chemical systems. For example, a CRN with the single reaction $X_1+X_2 \rxn Y$ produces one molecule of $Y$ whenever an $X_1$ reacts with an $X_2$. 
Interpreting the counts of $X_1,X_2$ as the vector-valued input and the count of $Y$ as the output, the CRN effectively calculates the function $f(x_1, x_2) = \min(x_1, x_2)$. Suppose the single reaction is instead $X_1+X_2 \rxn 2X_2$, where $X_2$ acts as a catalyst, producing another molecule of $X_2$ when it reacts with a $X_1$. Interpreting the presence of $X_1$ as ``no'' and the presence of $X_2$ as ``yes'', the CRN effectively calculates the predicate $\phi(x_1, x_2) = \text{yes} \Leftrightarrow x_2 > 0$.
% \todo{DD: explain example of predicate computation; note different notions of output in two types of computation}
It is known that precisely the \emph{semilinear} predicates $\phi : \mathbb{N}^k \to \{0, 1\}$ \cite{DBLP:journals/dc/AngluinADFP06} and functions $f : \mathbb{N}^k \to \mathbb{N}$ \cite{DBLP:journals/nc/ChenDS14} can be computed \emph{stably}, which roughly means that we can reach a correct state regardless of the reactions that occur or the order in which they occur.
More precisely,
we require that for every state $\vc$ reachable from the initial state encoding the input, there is a correct state $\vy$ reachable from $\vc$,
where $\vy$ is also \emph{stable},
meaning that every state reachable from $\vy$ has the same output as $\vy$.
Under the reasonable assumption that only a finite number of states are reachable from the initial state,
this simple combinatorial definition based on reachability is equivalent to requiring that the CRN \emph{will} reach a correct, stable state with probability 1.

Existing constructions \cite{DBLP:journals/dc/AngluinADFP06,DBLP:journals/nc/ChenDS14,DBLP:journals/nc/DotyH15} assume that reactions may not occur in reverse.
In reality, it is known that any chemical reaction that can happen in the forward direction, turning reactants into products, 
can also happen in reverse, turning the products back into the reactants, even if at a much slower rate.\footnote{
    The fully general model of CRNs associates to a reaction such as $A \revrxn^f_r B$ a forward \emph{rate constant} $f$ and reverse rate constant $r$.
    Without getting into details of the precise energy model,
    if $f > r$, this means that $B$ is more energetically stable than $A$,
    but to have $r=0$ (i.e., the reaction is completely irreversible) requires that $A$ be have energy \emph{infinitely} larger than $B$.
}
To address this incompatibility, we introduce the notion of \emph{reverse-robust} computation with CRNs,
in which reactions are allowed to occur in reverse, but we must always be able to reach a correct state using forward reactions.
We formalize this by generalizing the previous definition of stable computation to require that from any state $\vc$ reachable from the initial state by any combination of forward and reverse reactions,
there is a correct, stable state reachable using only forward reactions.
In other words, although the reverse reactions can take the CRN to states potentially unreachable through forward reactions alone, such reactions cannot prevent the CRN from reaching the correct answer.
We note that this definition allows reverse reactions only transiently; 
if they eventually stop, then forward reactions will always take the CRN to a correct, stable state.\footnote{
    Note that if reverse reactions are permitted to occur indefinitely,
    this would trivially prevent any state from being stable, since one could reverse out of any state possibly all the way back to the initial state.
}

To understand the distinction between this model and the stable computation model, consider the 
CRN $2X \rxn Y$, which stably computes the function $f(x) = \lfloor x/2 \rfloor$.
We then add two reactions involving a third species, $Z$:
\begin{align}
    \label{rxn:div_example_1}
    2X & \rxn Y \\
    \label{rxn:div_example_2}
    Z & \rxn Y \\
    \label{rxn:div_example_3}
    Z & \rxn \varnothing
\end{align}
Once again interpreting $X$ as the input and $Y$ as the output, under the stable computation model, since there are no copies of $Z$ present initially and no reaction produces $Z$, there is no way for reactions \eqref{rxn:div_example_2} and \eqref{rxn:div_example_3} to occur. 
Consequently, the CRN behaves identically to having only reaction \eqref{rxn:div_example_1},
and computes $f(x) = \lfloor x/2 \rfloor$. Consider, however, what happens when reactions are allowed to occur in reverse, under the reverse-robust computation model: when reaction \eqref{rxn:div_example_1} produces a $Y$, reaction \eqref{rxn:div_example_2} in reverse may convert the $Y$ to a $Z$, and reaction \eqref{rxn:div_example_3} (forward) may then annihilate that $Z$. 
It is then possible to reach a state where no molecules remain, from which no correct state can be reached by forward reactions.
Reactions \eqref{rxn:div_example_2} and \eqref{rxn:div_example_3} behave like a ``trap'', only disrupting the computation when reverse reactions are allowed.

% \todoi{DD: add paragraph explaining main results of the paper, and say a bit about how they are proven at a high level.

% - can compute exactly semilinear things

% - note (before or after) that clearly we can do no more than these since we are requiring stability against a strictly more powerful adversary than in stable computation, so main challenge is showing positive result

% - state a bit about how we prove existing (or small modifications to existing) constructions are reverse-robust (invariants preserved by reactions, thus preserved by reverse of reactions)

% }

We show that it is precisely the semilinear predicates and functions that can be computed reverse-robustly; that is, CRNs are equally as powerful under the reverse-robust computation model as under the stable computation model. Note that the capabilities of reverse-robustly computing CRNs are clearly limited by the established results about the capabilities of stably computing CRNs. We may think of stable computation as requiring that the CRN work against an adversary that may choose which reaction occurs at each intermediate state, only allowing reactions to occur in the forward direction. We can then think of reverse-robust computation as requiring that the CRN work against an adversary that may choose which reaction occurs at each intermediate state, allowing reactions to occur in both the forward and reverse directions. The adversary of the reverse-robust computation model is evidently more powerful than that of the stable computation model; the main challenge is then to show that against this more powerful adversary, we can compute anything that we could compute against the weaker adversary. In fact, we show that the existing constructions for CRNs that worked against the weaker adversary to compute semilinear predicates and functions continue to work against the stronger adversary. To prove the correctness of those constructions under the reverse-robust computation model, we use ``invariants'', quantities that are preserved by the reactions of a CRN (and thus, by the reverse of those reactions as well).

The paper is organized as follows. Section 2 formally defines what it means for a CRN to compute reverse-robustly (Definitions \ref{def:rrd} and \ref{def:rrc}) and introduces the concept of an invariant. Sections 3 and 4 contain the main positive results of the paper, and prove that the existing constructions used to decide semilinear sets and compute semilinear functions stably also work under the reverse-robust computation model. 

%% file: prelim.tex
\section{Preliminaries}
\input{prelim/notation}
\input{prelim/crns}
\input{prelim/stable-rr-computation}
\input{prelim/invariants}

%% file: prelim/notation.tex
\subsection{Notation}
Let $\mathbb{N}$ denote the nonnegative integers. For any finite set $\Lambda$, $\mathbb{N}^\Lambda$ denotes the set of functions $f : \Lambda \to \mathbb{N}$. Equivalently, $\vec{c} \in \mathbb{N}^\Lambda$ can be interpreted as a vector of $|\Lambda|$ nonnegative integers, where each element specifies the nonnegative integer count of an element of $\Lambda$. $\vec{c}(i)$ denotes the $i$-th coordinate of $\vec{c}$, and if $\vec{c}$ is indexed by elements of $\Lambda$, then $\vec{c}(Y)$ denotes the count of species $Y \in \Lambda$. We sometimes use multiset notation for such vectors, e.g., $\{A, 3C\}$ for the vector $\vec{c} = (1, 0, 3)$, assuming there are three species $A, B, C$, indexed in that order. For $\Sigma \subseteq \Lambda$, $\vec{i} |\Sigma$ denotes the restriction of $\vec{i}$ to $\Sigma$.

%\todo{DD: mention (see other papers of mine for wording) that we equivalently view a vector as a multiset and can therefore write set operations such as $\vec{p} \cap \vec{a} = \emptyset$ to express that $\vec{p}$ and $\vec{a}$ are disjoint.}
For two vectors $\vec{x}, \vec{y} \in \mathbb{R}^k$, we write $\vec{x} \geqq \vec{y}$ to indicate that $\vec{x}(i) \geq \vec{y}(i)$ for all $1 \leq i \leq n$, $\vec{x} \geq \vec{y}$ to indicate that $\vec{x} \geqq \vec{y}$ but $\vec{x} \neq \vec{y}$, and $\vec{x} > \vec{y}$ to indicate that $\vec{x}(i) > \vec{y}(i)$ for all $1 \leq i \leq n$. When $\vec{y} = \vec{0}$, we refer to $\vec{x}$ as nonnegative, semipositive, and positive, respectively. $\leqq, \leq, <$ are defined analogously.

%% file: prelim/crns.tex
\subsection{Chemical reaction networks}
A \emph{chemical reaction network} (CRN) is a pair $\mathcal{C} = (\Lambda, R)$ where $\Lambda$
% = \{\lambda_1, \ldots, \lambda_k\}$ 
is a finite set of \emph{chemical species} and 
$R$
%= \{r_1, \ldots, r_d\}$ 
is a finite set of (forward) reactions over $\Lambda$, where a \emph{(forward) reaction} is a pair $r = (\vec{a}, \vec{p}) \in \mathbb{N}^\Lambda \times \mathbb{N}^\Lambda$ where $\vec{a}$ indicates the \emph{reactants} and $\vec{p}$ the \emph{products}.
%When written as $\vec{r}$, we understand it to mean $\vec{r} = \vec{p} - \vec{a}$. 
For a reaction $r = (\vec{a}, \vec{p})$, we use $\rev{r} = (\vec{p}, \vec{a})$ to denote the corresponding \emph{reverse reaction}.
% Henceforth, we use \emph{reaction} when a statement is applicable to both forward and reverse reactions.

A \emph{configuration} $\vec{c} \in \mathbb{N}^\Lambda$ assigns nonnegative integer counts to every species $\lambda \in \Lambda$.
A reaction $r = (\vec{a}, \vec{p})$ is \emph{applicable} to a configuration $\vec{x}$ if $\vec{x} \geqq \vec{a}$; applying $r$ to $\vec{x}$ then results in the configuration $\vec{y} = \vec{x} - \vec{a} + \vec{p}$.
We say $r$ is \emph{reverse-applicable} to $\vec{x}$ if $\vec{x} \geqq \vec{p}$. Equivalently, we may say that $\rev{r}$ is applicable to $\vec{x}$.
We say $r$ is \emph{bi-applicable} to $\vec{x}$ if it is either applicable or reverse applicable to $\vec{x}$.
If some reaction $r$ is applicable or reverse-applicable to $\vec{x}$ and (reverse-)applying it results in the configuration $\vec{y}$, we write $\vec{x} \bireach^1 \vec{y}$. If it is known that $r$ was applicable, we write $\vec{x} \reach^1 \vec{y}$.
A \emph{bi-execution} $\mathcal{E}$ is a finite or infinite sequence of configurations $(\mathbf{c}_0, \mathbf{c}_1, \mathbf{c}_2, \ldots)$ such that for all $j \in \{1, \ldots, |\mathcal{E}| - 1\}$, $\mathbf{c}_{j-1} \bireach^1 \mathbf{c}_j$. A \emph{(forward-)execution} is defined analogously, instead requiring that $\mathbf{c}_{j-1} \reach \mathbf{c}_j$. A configuration $\mathbf{y}$ is \emph{bireachable} from another configuration $\mathbf{x}$ (written as $\vec{x} \bireach \vec{y}$) if there exists some finite bi-execution that starts at $\mathbf{x}$ and ends at $\mathbf{y}$, and \emph{reachable} (written as $\vec{x} \reach \vec{y}$) if there exists such a forward-execution. Note that both bireachability and reachability are additive: if $\vec{x} \bireach \vec{y}$, for all $\vec{z} \in \mathbb{N}^\Lambda$, $\vec{x} + \vec{z} \bireach \vec{y} + \vec{z}$, and analogously for when $\vec{x} \reach \vec{y}$. It also follows trivially that both bireachability and reachability are transitive; however, bireachability is clearly reflexive, as we can simply reverse-apply the reactions moving backwards through the sequence of configurations, while the same does not apply for reachability.

%% file: prelim/stable-rr-computation.tex
\subsection{Stable and reverse-robust computation with CRNs}
To perform computation using CRNs, we define a way to interpret the final configuration after letting forward reactions until the result can no longer change (characterized below as \emph{stable computation}). We define \emph{reverse-robust computation} analogously, letting both forward and reverse reactions occur until the result can no longer change (through forward reactions). Computation using CRNs primarily involves the evaluation of \emph{predicates} $\phi : \mathbb{N}^k \to \{0, 1\}$ and \emph{functions} $f : \mathbb{N}^k \to \mathbb{N}$.

The definitions below reference \emph{input species} $\Sigma \subseteq \Lambda$ and an \emph{initial context} $\vec{s} \in \mathbb{N}^{\Lambda \setminus \Sigma}$. If $\vec{s} = \vec{0}$, we say that the CRN is \emph{leaderless}. We may assume without loss of generality that if the initial context is nonzero, it contains a single copy of some \emph{leader species} $L$. In both cases, we say that $\vec{i} \in \mathbb{N}^\Lambda$ is a \emph{valid initial configuration} if $\vec{i} = \vec{s} + \vec{x}$ where $\vec{x}(\lambda_j) = 0$ for all $\lambda_j \in \Lambda \setminus \Sigma$; i.e., $\vec{i}$ contains only input species in addition to the initial context.

A \emph{chemical reaction decider} (CRD) is a tuple $\mathcal{D} = (\Lambda, R, \Sigma, \Upsilon_1, \Upsilon_0, \vec{s})$ where $(\Lambda, R)$ is a CRN, $\Sigma \subseteq \Lambda$ is the set of input species, $\Upsilon_1$ is the set of \emph{yes voters}, $\Upsilon_0$ is the set of \emph{no voters}, and $\vec{s} \in \mathbb{N}^{\Lambda \setminus \Sigma}$ is the initial context. In the context of CRDs, we use \emph{output species} to refer to $\Gamma = \Upsilon_0 \cup \Upsilon_1$, i.e. the set of voters. We define a global output partial function $\Phi : \mathbb{N}^\Lambda \to \{0, 1\}$ as follows: $\Phi(\vec{c}) = 0$ if $\vec{c}(\lambda_j) > 0$ for some $\lambda_j \in \Upsilon_0$ and $\vec{c}(\lambda_j') = 0$ for all $\lambda_j' \in \Upsilon_1$, $\Phi(\vec{c}) = 1$ if $\vec{c}(\lambda_j) > 1$ for some $\lambda_j \in \Upsilon_1$ and $\vec{c}(\lambda_j') = 0$ for all $\lambda_j' \in \Upsilon_0$, and is undefined otherwise. In other words, a unanimous vote is required for an output. We say a configuration $\vec{c}$ is \emph{stable} if for all $\vec{c}'$ such that $\vec{c} \reach \vec{c}'$, $\Phi(\vec{c}) = \Phi(\vec{c}')$. A CRD $\mathcal{D}$ is said to \emph{stably decide} the predicate $\phi : \mathbb{N}^\Sigma \to \{0, 1\}$ if for any valid initial configuration $\vec{i} \in \mathbb{N}^\Lambda$, letting $\vec{i}_0 = \vec{i} | \Sigma$, for all $\vec{c} \in \mathbb{N}^\Lambda$ such that $\vec{i} \reach \vec{c}$, there exists some $\vec{c}'$ such that $\vec{c} \reach \vec{o}$ where $\vec{o}$ is stable and $\Phi(\vec{o}) = \phi(\vec{i}_0)$.
\begin{definition}\label{def:rrd}
A CRD $\mathcal{D}$ is said to \emph{reverse-robustly decide} the predicate $\phi : \mathbb{N}^\Sigma \to \{0, 1\}$ if for any valid initial configuration $\vec{i} \in \mathbb{N}^\Lambda$, letting $\vec{i}_0 = \vec{i} | \Sigma$, for all $\vec{c} \in \mathbb{N}^\Lambda$ such that $\vec{i} \bireach \vec{c}$, there exists some $\vec{c}'$ such that $\vec{c} \reach \vec{o}$ where $\vec{o}$ is stable and $\Phi(\vec{o}) = \phi(\vec{i}_0)$.
\end{definition}

A \emph{chemical reaction computer} (CRC) is a tuple $\mathcal{C} = (\Lambda, R, \Sigma, Y, \vec{s})$ where $(\Lambda, R)$ is a CRN, $\Sigma \subseteq \Lambda$ is the set of input species, $Y \in \Lambda$ is the  \emph{output species}, and $\vec{s} \in \mathbb{N}^{\Lambda \setminus \Sigma}$ is the initial context. We say a configuration $\vec{c}$ is \emph{stable} if for all $\vec{c}'$ such that $\vec{c} \reach \vec{c}'$, $\vec{c}(Y) = \vec{c}'(Y)$. A CRC $\mathcal{C}$ is said to \emph{stably compute} the function $f : \mathbb{N}^\Sigma \to \mathbb{N}$ if for any valid initial configuration $\vec{i} \in \mathbb{N}^\Lambda$, letting $\vec{i}_0 = \vec{i} | \Sigma$, for all $\vec{c} \in \mathbb{N}^\Lambda$ such that $\vec{i} \reach \vec{c}$, there exists some $\vec{c}'$ such that $\vec{c} \reach \vec{o}$ where $\vec{o}$ is stable and $\vec{o}(Y) = f(\vec{i}_0)$.
\begin{definition}\label{def:rrc}
A CRC $\mathcal{C}$ is said to \emph{reverse-robustly compute} the function $f : \mathbb{N}^\Sigma \to \mathbb{N}$ if for any valid initial configuration $\vec{i} \in \mathbb{N}^\Lambda$, letting $\vec{i}_0 = \vec{i} | \Sigma$, for all $\vec{c} \in \mathbb{N}^\Lambda$ such that $\vec{i} \bireach \vec{c}$, there exists some $\vec{c}'$ such that $\vec{c} \reach \vec{o}$ where $\vec{o}$ is stable and $\vec{o}(Y) = f(\vec{i}_0)$.
\end{definition}

Henceforth, we use CRN when a statement is applicable to either CRCs or CRDs. We can then say that a CRN computes stably if for all valid input configurations $\vec{i}$, for all configurations $\vec{c}$ such that $\vec{i} \reach \vec{c}$, there exists some stable correct configuration $\vec{o}$ such that $\vec{c} \reach \vec{o}$. Analogously, a CRN computes reverse-robustly if for all valid input configurations $\vec{i}$, for all configurations $\vec{c}$ such that $\vec{i} \bireach \vec{c}$, there exists some stable correct configuration $\vec{o}$ such that $\vec{c} \reach \vec{o}$.

% Reverse-robust computation is defined analogously:
% \todo{DD: I'd prefer to separate the definitions of deciding predicates and computing functions.}
% \begin{definition}
% A CRN \emph{reverse-robustly computes} if for all valid input configurations $\vec{i}$, for all configurations $\vec{c}$ such that $\vec{i} \bireach \vec{c}$, there exists some stable correct configuration $\vec{o}$ such that $\vec{c} \reach \vec{o}$.
% \end{definition}

\subsection{``Shuffling'' executions}
The ability to parallelize computation is a key technique used in designing CRNs that compute stably.
By this we mean the common trick of splitting input species $X_i$ into two variants via a reaction $X_i \rxn X_{i,1} + X_{i,2}$, 
where $X_{i,j}$ is the $i$'th input species for CRN $C_j$, and $C_1$ and $C_2$ then run in parallel on their respective copies of the input species.
We show that this technique continues to work under the reverse-robust computation model.

We first prove a technical lemma about reachability that will help in showing this.
If $\vec{c}_0 \reachrxn{r_1} \vec{c}_1 \reachrxn{r_2} \vec{c}_2$,
then we cannot in general commute the reactions since $r_2$ may not be applicable to $\vec{c}_0$.
Furthermore, even if $r_2$ \emph{is} applicable to $\vec{c}$ and $\vec{c}_0 \reachrxn{r_2} \vec{c}_1'$,
then $r_1$ may not be applicable to $\vec{c}_1'$.
However, if none of the products of $r_1$ are reactants of $r_2$, then in fact the reactions can be commuted as shown next.

\begin{lemma}\label{lem:rxn-swap}
Let $\vec{c}_0, \vec{c}_1, \vec{c}_2$ such that $\vec{c}_0 \reachrxn{r_1} \vec{c}_1 \reachrxn{r_2} \vec{c}_2$,
where 
$r_1 = (\vec{a}_1, \vec{p}_1)$
and
$r_2 = (\vec{a}_2, \vec{p}_2)$.
If $\vec{p}_1 \cap \vec{a}_2 = \emptyset$ (none of the products of $r_1$ are reactants of $r_2$),
then there is some $\vec{c}_1'$ such that $\vec{c}_0 \reachrxn{r_2} \vec{c}_1' \reachrxn{r_1} \vec{c}_2$.
\end{lemma}

\begin{proof}
Since $\vec{c}_0 \reachrxn{r_1} \vec{c}_1$ and $\vec{c}_1 \reachrxn{r_2} \vec{c}_2$, $\vec{c}_0 \geqq \vec{a}_1$ and $\vec{c}_1 \geqq \vec{a}_2$.
If $\vec{p}_1(\lambda) > 0$ ($\lambda$ is a product of $r_1$), since we have assumed by hypothesis that is it not a reactant of $r_2$, 
then $\vec{a}_2(\lambda) = 0$. 
Consequently, $\vec{a}_1(\lambda) + \vec{a}_2(\lambda) = \vec{a}_1(\lambda)$, and so $\vec{c}_0(\lambda) \geq \vec{a}_1(\lambda) + \vec{a}_2(\lambda)$.
On the other hand, if $\lambda$ is not a product of $r_1$ ($\vec{p}_1(\lambda) = 0$)
\begin{align*}
    \vec{c}_0(\lambda) & = \vec{c}_1(\lambda) + \vec{a}_1(\lambda) - \vec{p}_1(\lambda) 
    \\&= \vec{c}_1(\lambda) + \vec{a}_1(\lambda) 
    && \text{since } \vec{p}_1(\lambda) = 0
    \\&\geq
    \vec{a}_2(\lambda) + \vec{a}_1(\lambda).
    && \text{since } \vec{c}_1(\lambda) \ge \vec{a}_2(\lambda)
\end{align*}
It then follows that $\vec{c}_0 \geqq \vec{a}_1 + \vec{a}_2$ implying $\vec{c}_0 \geqq \vec{a}_2$ since $\vec{a}_1 \geqq \vec{0}$.
Consequently, $r_2$ is applicable to $\vec{c}_0$; 
let $\vec{c}_1' = \vec{c}_0 - \vec{a}_2 + \vec{p}_2$ be the resulting configuration. 
It remains to show $r_1$ is applicable to $\vec{c}_1',$
i.e., $\vec{c}_1' \geqq \vec{a}_1.$
We have
\begin{align*}
    \vec{c}_1' 
    &=
    \vec{c}_0 - \vec{a}_2 + \vec{p}_2
    \\&\geqq
    (\vec{a}_1 + \vec{a}_2) - \vec{a}_2 + \vec{p}_2
    &&
    \text{since }
    \vec{c}_0 \geqq \vec{a}_1 + \vec{a}_2
    \\&=
    \vec{a}_1 + \vec{p}_2 \geqq \vec{a}_1.
    &&
    \text{since } 
    \vec{p}_2 \geqq \vec{0} \qedhere
\end{align*}
\end{proof}
By a similar argument, we may commute forward reactions with reverse reactions, provided that none of the products of the reaction that occurs first are reactants of the reaction that occurs second.

Intuitively, when a forward reaction occurs consecutively with its reverse reaction, the pair may be canceled out, producing a bi-execution where those reactions do not occur.
% \todo{DD: make sure the previous sentence belongs here; I'm still not sure this belongs here.}
We now show that duplicating the input species allows for the parallel reverse-robust computation of multiple predicates or functions.

% \todo{DD: need to formalize what we mean by what the resulting CRN is doing.}
\begin{lemma}\label{lem:parallel-comp}
Let $\mathcal{C}_1, \mathcal{C}_2$ be some reverse-robustly computing CRNs with input species $\{X_1, \ldots, X_k\}$. For clarity, rename $X_i$ to $X_{i,1}$ in $\mathcal{C}_1$ and $X_{i,2}$ and $\mathcal{C}_2$. Let $\mathcal{C}$ be the CRN formed by combining the reactions of $\mathcal{C}_1$ and $\mathcal{C}_2$ and adding the reaction $r_{i,0}$, for each $i \in \{1, \ldots, k\}$:
\begin{align}\label{rxn:parallel-split}
    X_i \rxn X_{i,1} + X_{i,2}
\end{align}
Then for all valid input configurations $\vec{i}$, for all configurations $\vec{c}$ such that $\vec{i} \bireach \vec{c}$, there exists some stable configuration $\vec{o}$ such that $\vec{c} \reach \vec{o}$ and the individual outputs of both $\mathcal{C}_1$ and $\mathcal{C}_2$ are correct.
\end{lemma}

\begin{proof}
Let $\vec{i}$ be an arbitrary initial configuration, where $\vec{i}(X_i) = x_i$. It is easy to see that if we start with only $x_i$ each of $X_{i,1}$ and $X_{i,2}$ in configuration $\vec{i}^*$, the CRNs can be run independently, and are able to reverse-robustly compute their respective functions/predicates in parallel. Then for any $\vec{c}^*$ such that $\vec{i}^* \bireach \vec{c}^*$, there exists some stable $\vec{o}^*$ such that $\vec{c}^* \reach \vec{o}^*$ and the individual outputs are correct.

Now let $\vec{c}$ such that $\vec{i} \bireach \vec{c}$. Note that by repeatedly applying $r_{i,0}$, we can reach some $\vec{c}'$ such that $\vec{c} \reach \vec{c}'$ and $r_{i,0}$ is no longer applicable; we may then assume that in $\vec{c}$, $r_{i,0}$ are no longer applicable, and so $X_i$ are not present.

We first argue that $\vec{i} \bireach \vec{c}$ by an execution where $r_{i,0}$ is never applied in reverse for all $i$. Let $\mathcal{E} = (\vec{c}_0, \vec{c}_1, \ldots, \vec{c}_n)$ be an execution where $r_{i,0}$ is reverse-applied (with $\vec{c}_0 = \vec{i}$ and $\vec{c}_n = \vec{c}$) for some $i$. Note that the only reaction where the products of $\rev{r}_{i,0}$ are reactants is $r_{i,0}$, and that those are the only reactions in which $X_i$ is involved; this affords us a great deal of flexibility in shuffling the order of reactions.

Note that we can always select $j$ and $k$ where $j < k$ such that $\vec{c}_j \bireach^{\rev{r}_{i,0}} \vec{c}_{j+1}$ and $\vec{c}_k \reach^{r_{i,0}} \vec{c}_{k+1}$ and $\vec{c}_{j+1} \bireach \vec{c}_k$ by a bi-execution where $r_{i,0}$ does not occur in forward or in reverse. This follows from the fact that $X_i$ is not present in $\vec{c}$; if $\rev{r}_{i,0}$ occurred, $r_{i,0}$ must have occurred afterwards, as it is the only way $X_i$ can be consumed.

We then have that:
\begin{align*}
    \vec{c}_j \bireach^{\rev{r}_{i,0}} \vec{c}_{j+1} \bireach \vec{c}_k \reach^{r_{i,0}} \vec{c}_{k+1}
\end{align*}
where $\vec{c}_{j+1} \bireach \vec{c}_k$ by a bi-execution where $r_{i,0}$ does not occur in forward or in reverse. Since none of the products of $\rev{r}_{i,0}$ ($X_i$) are reactants of any of the reactions that occur between $\vec{c}_{j+1}$ and $\vec{c}_k$, we can ``shuffle'' $\rev{r}_{i,0}$ forward using Lemma \ref{lem:rxn-swap} such that:
\begin{align*}
    \vec{c}_j \bireach \vec{c}_{k-1}' \bireach^{\rev{r}_{i,0}} \vec{c}_k \reach^{r_{i,0}} \vec{c}_{k+1}
\end{align*}
where $\vec{c}_{j} \bireach \vec{c}_{k-1}'$ by a bi-execution where $r_{i,0}$ does not occur in forward or in reverse. We may then cancel out $\rev{r}_{i,0}$ and $r_{i,0}$, and so $\vec{c}_{j} \bireach \vec{c}_{k+1}$ by a bi-execution where $r_{i,0}$ does not occur in forward or in reverse.

By repeatedly canceling out any reverse-applications of $r_{i,0}$ in this manner, we see that $\vec{i} \bireach \vec{c}$ by an execution where $\rev{r}_{i,0}$ never occurs for all $i$.

We now further argue that $\vec{i} \bireach \vec{c}$ by an execution where $r_{i,0}$ is never reverse-applied for all $i$ and $\vec{i}^*$ is reached at some point. Since we have an execution that starts at $\vec{i}$ and ends at $\vec{c}$ where $\rev{r}_{i,0}$ never occurs for all $i$, no reactions that have the reactants of $r_{i,0}$ ($X_i$) as products occur for all $i$. This allows us to shuffle the applications of $r_{i,0}$ to the beginning of the execution using Lemma \ref{lem:rxn-swap}:
\begin{align*}
    \vec{i} \reach \vec{i}' \bireach \vec{c}
\end{align*}
where $\vec{i} \reach \vec{i}'$ by a forward-execution where only $r_{i,0}$ are applied and $\vec{i}' \bireach \vec{c}$ by a bi-execution where $r_{i,0}$ does not occur in forward or in reverse for all $i$. It is easy to see that $\vec{i}' = \vec{i}^*$, and so it follows that there exists some stable $\vec{o}$ such that $\vec{c} \reach \vec{o}$ and the individual outputs are correct. It follows that $\vec{c} \reach \vec{o}$. Consequently, for all valid input configurations $\vec{i}$, for all configurations $\vec{c}$ such that $\vec{i} \bireach \vec{c}$, there exists some stable configuration $\vec{o}$ such that $\vec{c} \reach \vec{o}$ and the individual outputs of both $\mathcal{C}_1$ and $\mathcal{C}_2$ are correct.
\end{proof}

%% file: prelim/invariants.tex
\subsection{Invariants}
For an arbitrary CRN, we define an invariant to be some function $I : \mathbb{N}^\Lambda \to \mathbb{N}$ such that for all reactions $r = (\vec{a}, \vec{p})$, $I(-\vec{a}+\vec{p}) = 0$. Note that if $I$ is an invariant, so is $kI$ for all $k \in \mathbb{N}$ and if $I_1, I_2$ are invariants, so is $I_1 + I_2$. This paper uses what we refer to as linear and modular invariants:
\begin{itemize}
    \item An invariant is \emph{linear} if it is of the form $I_{\text{linear}}(-\vec{a}+\vec{p}) = \sum_{j=1}^k w_j(-\vec{a}(\lambda_j)+\vec{p}(\lambda_j))$ for $w_1, \ldots, w_k \in \mathbb{Z}$.
    \item An invariant is \emph{modular} if it is of the form \\$I_{\text{modular}}(-\vec{a}+\vec{p}) = \left(\sum_{j=1}^k w_j(-\vec{a}(\lambda_j)+\vec{p}(\lambda_j))\right) \bmod m$ for $w_1, \ldots, w_k \in \mathbb{Z}$ and $m \in \mathbb{N}$. We refer to $m$ as the \emph{modulus} of the invariant.
\end{itemize}
Further, we refer to a linear function as the \emph{linearization} of a modular invariant if taking the linear function modulo $m$ yields the modular invariant. For example, in the following CRN:
\begin{align}
    2A_0 & \rxn A_0 \\
    A_0 + A_1 & \rxn A_1\\
    2A_1 & \rxn A_0\\
    2B & \rxn B + C
\end{align}
Examining reactions (5)-(7), note that when $A_i$ reacts with $A_j$, the product is $A_{i + j \bmod 2}$. Similarly, examining reaction (8), note that $B$ acts as a catalyst to turn copies of $B$ into copies of $C$. We then see that $I_1(-\vec{a}+\vec{p}) = \left(\sum_{i=0}^1 i(-\vec{a}(A_i) + \vec{p}(A_i))\right) \bmod 2$ is a modular invariant and that $I_2(-\vec{a}+\vec{p}) = -\vec{a}(B) + \vec{p}(B) - \vec{a}(C) + \vec{p}(C)$ is a linear invariant. The linearization of $I_1$ is $I_{1,\text{lin}} = \sum_{i=0}^1 i(-\vec{a}(A_i) + \vec{p}(A_i))$.

\begin{observation}
For a linear or modular invariant $I$, it is equivalent to state that for an initial configuration $\vec{i}$ and all $\vec{c}$ such that $\vec{i} \bireach \vec{c}$, $I(\vec{c})$ is constant and equal to $I(\vec{i})$.
\end{observation}
To see why this is true, let $\vec{c}$ be an arbitrary configuration, and $r = (\vec{a}, \vec{p})$ be an arbitrary reaction. If $I$ is a linear invariant:
\begin{align*}
    I(\vec{c} - \vec{a} + \vec{p}) = I(\vec{c}) + I(-\vec{a} + \vec{p}) = I(\vec{c})
\end{align*}
If $I$ is a modular invariant with modulus $m$:
\begin{align*}
    I(\vec{c} - \vec{a} + \vec{p}) & = I_{\text{lin}}(\vec{c} - \vec{a} + \vec{p}) \bmod m \\
    & = (I_{\text{lin}}(\vec{c}) + I_{\text{lin}}( - \vec{a} + \vec{p})) \bmod m \\
    & = I_{\text{lin}}(\vec{c}) \bmod m = I(\vec{c})
\end{align*}

%% file: results.tex
\section{Exactly the semilinear sets are reverse-robustly decidable}

% \todo{DD: add definitions of reachable and bi-reachable, and introduce relation symbols for them to make the proofs less wordy.}

In this section, we will show that under the reverse-robust computation model, CRDs have the same computational power as under the stable computation model. The following is the main result of the section:

\begin{theorem}\label{thm:semilinear-rr-decidable}
Exactly the semilinear sets can be reverse-robustly decided by a CRD.
\end{theorem}

Since semilinear sets are Boolean combinations of mod and threshold predicates, we prove this theorem by showing that both mod and threshold sets can be reverse-robustly decided, as well as any Boolean combination of such.

\input{results/mod_set_decidable}
\input{results/threshold_set_decidable}
\input{results/boolean_closure}

\section{Exactly the semilinear functions are reverse-robustly computable}
In this section, we will show that under the reverse-robust computation model, CRCs have the same computational power as under the stable computation model. The following is the main result of the section:

\begin{theorem}
Exactly the semilinear functions can be reverse-robustly computed by a CRC.
\end{theorem}

\input{results/partial_affine_computable}
\input{results/semilinear_computable}

%% file: results/mod_set_decidable.tex
\subsection{All mod sets can be reverse-robustly decided}
\begin{lemma}\label{lem:mod-set}
Every mod set $M = \{(x_1, \ldots, x_k) \} \mid \sum_{i=1}^n w_ix_i \equiv c \bmod m\}$ is reverse-robustly decidable by a CRD.
\end{lemma}
We employ a modified version of the construction from \cite{DBLP:journals/dc/AngluinADFP06}. Let the set of input species be $\Sigma = \{X_1, \ldots, X_k\}$. We then add the following reaction, for each $i \in \{1, \ldots, k\}$:
\begin{align}
    X_i \rxn Y_{w_i \bmod m}
    \label{rxn:mod-set-conversion}
\end{align}
Then, for each $p, q \in \{0, \ldots, m - 1\}$ add the reactions:
\begin{align}
    Y_p + Y_q \rxn Y_{p + q \bmod m}
    \label{rxn:mod-set-addition}
\end{align}
Let the set of yes voters be $\Gamma_1 = \{Y_c\}$ and the set of no voters be $\Gamma_0 = \{Y_p \mid p \neq c\}$.
\begin{proof}
By inspecting the reactions, we see that the CRN has the modular invariant:
\begin{align*}
    I_M(\vec{c}) = \left(\sum_{i=0}^k w_i\vec{c}(X_i) + \sum_{p=0}^{m - 1} p\vec{c}(Y_p)\right) \bmod m
\end{align*}
For an initial configuration $\vec{i}$ where $\vec{i}(X_i) = x_i$, the invariant has value:
\begin{align*}
    I_M(\vec{i}) = \sum_{i=1}^k w_i\vec{c}(X_i) \bmod m = \sum_{i=1}^d w_ix_i \bmod m
\end{align*}

We now prove that this construction reverse-robustly decides $M$. Let $\vec{i}$ be an arbitrary initial configuration, where $\vec{i}(X_i) = x_i$, and let $\vec{c}$ such that $\vec{i} \bireach \vec{c}$. If any $X_i$ are present in $\vec{c}$, we first apply reaction (\ref{rxn:mod-set-conversion}) repeatedly until $X_i$ is no longer present. We then apply reaction (\ref{rxn:mod-set-addition}) repeatedly until only one $Y_{p^*}$ is present and no more reactions can occur; the resulting configuration $\vec{o}$ must then be stable. To see that $\vec{o}$ also has the correct output, note that the invariant has value:
\begin{align*}
    I_M(\vec{o}) = p^*\vec{o}(Y_{p^*}) \bmod m = p^*
\end{align*}
The remaining $Y_{p^*}$ is then a yes voter if and only if $\sum_{i=1}^d w_ix_i \equiv c \bmod m$. Since $\vec{o}$ is stable with the correct output and $\vec{c} \reach \vec{o}$, the CRD reverse-robustly decides $M$.
\end{proof}

%% file: results/threshold_set_decidable.tex
\subsection{All threshold sets can be reverse-robustly decided}
\begin{lemma}\label{set:threshold-set}
Every threshold set $T = \{(x_1, \ldots, x_k) \} \mid \sum_{i=1}^k w_ix_i \geq t\}$ is reverse-robustly decidable by a CRD.
\end{lemma}
We employ a modified version of the construction from \cite{DBLP:journals/dc/AngluinADFP06}. Let the set of input species be $\Sigma = \{X_1, \ldots, X_k\}$ and define $m$ such that $c = \max\{|w_1|, \ldots, |w_k|, |t|\} + 1$. We then add the following reaction, for each $i \in \{1, \ldots, k\}$:
\begin{align}
    X_i & \rxn Y^L_{w_i}
    \label{rxn:thresh-conversion}
\end{align}
The, for each $p, q \in \{-c, \ldots, c\}$ add the reactions:
\begin{align}
    Y^L_p + Y^?_q & \rxn \begin{cases}
        Y^L_{-c} + Y^F_{(p + q) + c} & p + q < -c \\
        Y^L_{p + q} & -c \leq p + q \leq c \\
        Y^L_{c} + Y^F_{(p + q) - c} & c < p + q \\
    \end{cases}
    \label{rxn:thresh-addition}
\end{align}
Let the set of yes voters be $\Gamma_1 = \{Y^L_p \mid p \geq t\}$ and the set of no voters be $\Gamma_0 = \{Y^L_p \mid p < t\}$.
\begin{proof}
By inspecting the reactions, we see that the CRN has the linear invariant:
\begin{align*}
    I_T(\vec{c}) = \sum_{i=0}^k w_i\vec{c}(X_i) + \sum_{p=-c}^c p\vec{c}(Y^L_p)  + \sum_{q=-c}^c q\vec{c}(Y^F_q)
\end{align*}
For an initial configuration $\vec{i}$ where $\vec{i}(X_i) = x_i$, the invariant has value:
\begin{align*}
    I_T(\vec{i}) & = w_i\vec{c}(X_i) = \sum_{i=1}^k w_ix_i
\end{align*}

We now prove that this construction reverse-robustly decides $T$. Let $\vec{i}$ be an arbitrary initial configuration, where $\vec{i}(X_i) = x_i$, and let $\vec{c}$ such that $\vec{i} \bireach \vec{c}$. If any $X_i$ are present in $\vec{c}$, we first apply reaction (\ref{rxn:thresh-conversion}) repeatedly until $X_i$ is no longer present. We then apply reaction (\ref{rxn:thresh-addition}) repeatedly until only one $Y^L_{p^*}$ is present (possibly in addition to some $Y^F_q$) and no more reactions can occur; the resulting configuration $\vec{o}$ must then be stable. To see that $\vec{o}$ also has the correct output, we consider two cases.

In the case that $Y^L_{p^*}$ is the only remaining molecule, the invariant has value:
\begin{align*}
    I_T(\vec{o}) = p^*\vec{o}(Y^L_{p^*}) = p^*
\end{align*}
The remaining $Y^L_{p^*}$ is then a yes voter if and only if $\sum_{i=1}^k w_ix_i \geq t$.

Alternatively, if some $Y^F_{q}$ are present, we first note that either $p^* = -c$ or $p^* = c$ and further, $p^*$ and all $q$ must have the same sign; otherwise reaction (\ref{rxn:thresh-addition}) would be applicable. If $p^* = c > 0$, the invariant has value:
\begin{align*}
    I_T(\vec{o}) & = c\vec{o}(Y^L_{c}) + \sum_{q=1}^c q\vec{o}(Y^F_{q^*}) = c + \sum_{q=1}^c q\vec{o}(Y^F_{q^*})
\end{align*}
The remaining $Y^L_{c}$ is a yes voter, which is correct since $\sum_{i=1}^k w_ix_i \geq t = c + \sum_{q=1}^c q\vec{o}(Y^F_{q^*}) > c$. By a similar argument, if $p^* = -c$, the remaining $Y^L_{-c}$ is correctly a no voter.

Since $\vec{o}$ is stable with the correct output and $\vec{c} \reach \vec{o}$ in all cases, the CRD reverse-robustly decides $T$.
\end{proof}

%% file: results/boolean_closure.tex
\subsection{Reverse-robustly decidable sets are closed under Boolean operations}
\begin{lemma}\label{lem:boolean-closure}
If sets $X_1, X_2 \subseteq \mathbb{N}^d$ are reverse-robustly decided by some CRDs $\mathcal{D}_1$ and $\mathcal{D}_2$, then so are $X_1 \cup X_2$, $X_1 \cap X_2$, and $\overline{X_1}$.
\end{lemma}
\begin{proof}
To reverse-robustly decide $\overline{X_1}$, swap the yes and no voters. As the reactions themselves remain unchanged, it is trivial to show that this construction is reverse-robust.

To show that $X_1 \cup X_2$ and $X_1 \cap X_2$ are reverse-robustly decidable, consider a construction where both sets are decided in parallel and their votes are recorded in a new voter species. We first add reactions that duplicate the input and produce one of the four new voter species $V_{NN}, V_{NY}, V_{YN}, V_{YY}$:
\begin{align}
X_i \rxn X_{i,1} + X_{i,2} + V_{NN}
\label{rxn:boolean-closure-conversion}
\end{align}
Note that $V_{NN}$ is chosen arbitrarily. $X_{i,1}$ and $X_{i,2}$ then replace $X_i$ in any reactions involving $X_i$ in $\mathcal{D}_1$ and $\mathcal{D}_2$, respectively. We then add reactions that store the separate votes in the new voter species. The votes of the $\mathcal{D}_1$ and $\mathcal{D}_2$ are stored in the first and second subscript of the new voter species, respectively. For $b \in \{Y, N\}$ and for every $S_b$ and $T_b$ that vote in $\mathcal{D}_1$ and $\mathcal{D}_2$, respectively, we add the reactions:
\begin{align}
    S_b + V_{\overline{b}?} \rxn S_b + V_{b?}
    \label{rxn:boolean-closure-crd-flip-first-vote} \\
    T_b + V_{?\overline{b}} \rxn S_b + V_{?b}
    \label{rxn:boolean-closure-crd-flip-second-vote}
\end{align}
To reverse-robustly decide $X_1 \cup X_2$, we let the yes voters be $\Gamma = \{V_{NY}, V_{YN}, V_{YY}\}$ and to reverse-robustly decide $X_1 \cap X_2$, we let the yes voters be $\Gamma = \{V_{YY}\}$.

We now prove that this construction reverse-robustly decides $X_1 \cup X_2$ or $X_1 \cap X_2$. Let $\vec{i}$ be an arbitrary initial configuration, where $\vec{i}(X_i) = x_i$, and let $\vec{c}$ such that $\vec{i} \bireach \vec{c}$. If any $X_i$ are present in $\vec{c}$, we first apply reactions (\ref{rxn:boolean-closure-conversion}) repeatedly until $X_i$ is no longer present. We then run $\mathcal{D}_1$ and $\mathcal{D}_2$ until no more of their reactions can occur, then finally do the same for reactions (\ref{rxn:boolean-closure-crd-flip-first-vote}) and (\ref{rxn:boolean-closure-crd-flip-second-vote}); this resulting configuration $\vec{o}$ must then be stable. Since $\mathcal{D}_1$ and $\mathcal{D}_2$ reverse-robustly decide $X_1$ and $X_2$ and their voters do not have their counts changed by reactions (\ref{rxn:boolean-closure-crd-flip-first-vote}) and (\ref{rxn:boolean-closure-crd-flip-second-vote}), by Lemma \ref{lem:parallel-comp}, their individual votes are stable and correct in $\vec{o}$. If any of the new voter molecules with the wrong overall vote are present, either (\ref{rxn:boolean-closure-crd-flip-first-vote}) and (\ref{rxn:boolean-closure-crd-flip-second-vote}) would be applicable; all of the new voter molecules must then have the correct overall vote. Since $\vec{o}$ is stable with the correct output, the CRD reverse-robustly decides $X_1 \cup X_2$ or $X_1 \cap X_2$, depending on how the yes and no voters are selected. 
\end{proof}

%% file: results/partial_affine_computable.tex
\subsection{Diff-representations of all partial affine functions can be reverse-robustly computed}
We say that a partial function $f : \mathbb{N}^k \to \mathbb{N}$ is \emph{affine} if there exist vectors $a_1, \ldots, a_k \in \mathbb{Q}$, $c_1, \ldots c_k \in \mathbb{N}$, and $b \in \mathbb{N}$ such that when $\vec{x}(i) - c_i \geq 0$:
\begin{align*}
    f(\vec{x}) = b + \sum_{i=1}^k a_i(\vec{x}(i) - c_i)
\end{align*}
For a partial function $f$ we write $\mathrm{dom} f$ for the domain of $f$, the set of inputs for which $f$ is defined. For convenience, we can work with integer valued molecule counts by multiplying by $\frac{1}{d}$ after the dot product, where $d$ may be taken to be the least common multiple of the denominators of the rational coefficients in the original definition such that $n_i = d \cdot a_i$:
\begin{align*}
    f(\vec{x}) = b + \sum_{i=1}^k a_i(\vec{x}(i) - c_i)
    \Longleftrightarrow
    f(\vec{x}) = b + \frac{1}{d}\sum_{i=1}^k n_i(\vec{x}(i) - c_i)
\end{align*}
We say that a partial function $\hat{f} : \mathbb{N}^k \to \mathbb{N}^2$ is a \emph{diff-representation} of $f$ if $\mathrm{dom} f = \mathrm{dom} \hat{f}$ and for all $x \in \mathrm{dom} f$, if $(y_P, y_C) = \hat{f}(\vec{x})$, then $f(\vec{x}) = y_P - y_C$, and $y_P = O(f(x))$. In other words, $\hat{f}$ represents $f$ as the difference of its two outputs $y_P$ and $y_C$ , the larger output $y_P$ possibly being larger than the output of the original function, but at most a multiplicative constant larger.

\begin{lemma}\label{lem:affine-diff-representation}
    Let $f : \mathbb{N}^k \to \mathbb{N}$ be an affine partial function. Then there is a diff-representation $\hat{f} : \mathbb{N}^k \to \mathbb{N}^2$ and a CRC that monotonically reverse-robustly computes $\hat{f}$.
\end{lemma}
We employ the construction from \cite{DBLP:journals/nc/DotyH15}. Let the set of input species be $\Sigma = \{X_1, \ldots, X_k\}$ and the set of output species be $\Gamma = \{Y_P, Y_C\}$. Then, for each $i \in \{1, \ldots, k\}$, we add the following reaction:
\begin{align}
    X_i \rxn C_{i,1} + B + bY^P
    \label{rxn:affine-input}
\end{align}
Then, for each $i \in \{1, \ldots, k\}$ and for each $p, q \in \{1, \ldots, c_i\}$, add the reactions:
\begin{align}
    C_{i,p} + C_{i,q} \rxn \begin{cases}
        C_{i, p + q} & p + q \leq c_i \\
        C_{i, c_i} + (p + q - c_i)X_i' & p + q > c_i
    \end{cases}
    \label{rxn:affine-c-offset}
\end{align}
Then, for each $i \in \{1, \ldots, k\}$, add the reaction:
\begin{align}
    X_i' \rxn \begin{cases}
        n_iD_1^P & n_i \geq 0 \\
        (-n_i)D_1^C & n_i < 0 \\
    \end{cases}
    \label{rxn:affine-n}
\end{align}
Then, for each $p, q \in \{1, \ldots, d - 1\}$, add the reactions:
\begin{align}
    D_{p}^P + D_{q}^P & \rxn \begin{cases}
        D_{p + q}^P & p + q \leq d - 1 \\
        D_{p + q - d}^P + Y^P & p + q > d - 1
    \end{cases}
    \label{rxn:affine-dp} \\
    D_{p}^C + D_{q}^C & \rxn \begin{cases}
        D_{p + q}^C & p + q \leq d - 1 \\
        D_{p + q - d}^C + Y^C & p + q > d - 1
    \end{cases}
    \label{rxn:affine-dc}
\end{align}
Then, add the reaction:
\begin{align}
    B + B \rxn B + bY^C
    \label{rxn:affine-b}
\end{align}
\begin{proof}
By inspecting the reactions, we see that the CRN has the linear invariant:
\begin{align*}
    I_{\hat{f}}(\vec{c}) & = \sum_{i=1}^k n_i\left(\vec{c}(X_i) + \sum_{p=1}^{c_i} p\vec{c}(C_{i,p}) + \sum_{i=1}^k \vec{c}(X_i')\right) \\
    & + \sum_{p=1}^{d-1} p\vec{c}(D^P_p) - \sum_{q=1}^{d-1} q\vec{c}(D^C_q) + d(\vec{c}(Y^P) - \vec{c}(Y^C)) - bd\vec{c}(B)
\end{align*}
in addition to the modular invariants:
\begin{align*}
    I_{P}(\vec{c}) & = \left(\sum_{n_i\geq 0} n_i\left(\vec{c}(X_i) + \sum_{p=1}^{c_i} p\vec{c}(C_{i,p}) + \sum_{i=1}^k \vec{c}(X_i')\right) + \sum_{p=1}^{d-1} p\vec{c}(D^P_p)\right) \bmod d \\
    I_{C}(\vec{c}) & = \left(\sum_{n_i< 0} n_i\left(\vec{c}(X_i) + \sum_{p=1}^{c_i} p\vec{c}(C_{i,p}) + \sum_{i=1}^k \vec{c}(X_i')\right) - \sum_{q=1}^{d-1} q\vec{c}(D^C_q)\right) \bmod d
\end{align*}
For an initial configuration $\vec{i}$ where $\vec{i}(X_i) = x_i$, the invariants have values:
\begin{align*}
    I_{\hat{f}}(\vec{i}) & = \sum_{i=1}^k n_i\vec{c}(X_i) = \sum_{i=1}^k n_ix_i \\
    I_{P}(\vec{i}) & = \left(\sum_{n_i\geq 0} n_i\vec{i}(X_i)\right) \bmod d = \left(\sum_{n_i\geq 0} n_ix_i\right) \bmod d \\
    I_{C}(\vec{i}) & = \left(\sum_{n_i< 0} n_i\vec{i}(X_i)\right) \bmod d = \left(\sum_{n_i< 0} n_ix_i\right) \bmod d
\end{align*}

We now prove that this construction reverse-robustly decides $\hat{f}$. Let $\vec{i}$ be an arbitrary initial configuration, where $\vec{i}(X_i) = x_i$ and $x_i - c_i \geq 0$, and let $\vec{c}$ such that $\vec{i} \bireach \vec{c}$. If any $X_i$ are present in $\vec{c}$, we first apply reactions (\ref{rxn:affine-input}) repeatedly. We then apply reaction (\ref{rxn:affine-c-offset}) repeatedly until only one $C_{i,s_i^*}$ is present for each $i$ and no more of that reaction can occur. We then do the same for reaction (\ref{rxn:affine-n}) until no more $X_i'$ are present, and then reactions (\ref{rxn:affine-dp}) and (\ref{rxn:affine-dc}) until only one $D^P_{p^*}$ and $D^C_{q^*}$ are present. Finally, we do the same for reaction (\ref{rxn:affine-b}) until only one $B$ is present; this resulting configuration $\vec{o}$ must then be stable. To see that $\vec{o}$ also has the correct output, we first consider the following quantity:
\begin{align*}
    I_{c,i}(\vec{c}) = \vec{c}(X_i) + \sum_{p=1}^{c_i} p\vec{c}(C_{i,p})
\end{align*}
In the initial configuration, the quantity has value $I_{c,i}(\vec{i}) = \vec{i}(X_i) = x_i \geq c_i$. By inspecting reaction (\ref{rxn:affine-input}), as well as reactions (\ref{rxn:affine-n}) through (\ref{rxn:affine-b}), we see that those reactions preserve the quantity. For reaction (\ref{rxn:affine-c-offset}), however, while the quantity is preserved whenever $p + q \leq c_i$, this is not true whenever $p + q > c_i$, as $\sum_{p=1}^{c_i} p\vec{c}(C_{i,p})$ changes. Specifically, it decreases when the reaction is applied in forward and increases when applied in reverse. Even so, since $C_{i,c_i}$ is one of the products of the forward reaction, we see that the quantity remains at least $c_i$ for any $\vec{c}$ such that $\vec{i} \bireach \vec{c}$. Since only one $C_{i,s_i^*}$ is present for each $i$ and no $X_i$ are present, the quantity has the value $I_{c,i}(\vec{o}) = s_i^*\vec{o}(C_{i,s_i^*}) = s_i^*$, and it must be then be the case that $s_i^* = c_i$.

We next consider the modular variants $I_P$ and $I_C$. In $\vec{o}$, $I_P$ has the value:
\begin{align*}
    I_{P}(\vec{o}) & = \left(\sum_{n_i\geq 0}^k n_ic_i\vec{c}(C_{i,c_i}) + p^*\vec{c}(D^P_{p^*})\right) \bmod d = \left(\sum_{n_i\geq 0}^k n_ic_i + p^*\right) \bmod d
\end{align*}
For this to be true, $p^* = \sum_{n_i \geq 0} n_i(x_i - c_i) \bmod d$.  By a similar argument, $q^* = \sum_{n_i < 0} (-n_i)(x_i - c_i) \bmod d$. Note that for $f(\vec{x})$ to be an integer, $\sum_{i=1}^k n_i(\vec{x}(i) - c_i) \bmod d = (p^* - q^*) \bmod d = 0$. Since $0 \leq p^*, q^* < d$, it must be the case that $p^* - q^* = 0$. We then see that $I_{\hat{f}}$ has value:
\begin{align*}
    I_{\hat{f}}(\vec{o}) & = \sum_{i=1}^k n_i c_i\vec{o}(C_{i,c_i}) + p^*\vec{o}(D^P_{p^*}) - q^*\vec{o}(D^C_{q^*}) + d(\vec{o}(Y^P) - \vec{o}(Y^C)) - bd\vec{o}(B) \\
    & = \sum_{i=1}^k n_i c_i + p^* - q^* + d(\vec{o}(Y^P) - \vec{o}(Y^C)) - bd \\
    & = \sum_{i=1}^k n_i c_i + d(\vec{o}(Y^P) - \vec{o}(Y^C)) - bd
\end{align*}
For this to be true, it must be the case that:
\begin{align*}
    \vec{o}(Y^P) - \vec{o}(Y^C) & = b + \frac{1}{d}\sum_{i=1}^k n_i(x_i - c_i) = f(\vec{x})
\end{align*}
Since $\vec{o}(Y^P) - \vec{o}(Y^C) = f(\vec{x})$, the CRC reverse-robustly computes $\hat{f}$.
\end{proof}

%% file: results/semilinear_computable.tex
\subsection{All semilinear functions can be reverse-robustly computed}
We require the following result from \cite{DBLP:conf/wdag/DotyH24}, which guarantees that any semilinear function can be decomposed into affine partial functions with disjoint domains.

\begin{lemma}\label{lem:decompose_semilinear}
Let $f : \mathbb{N}^k \to \mathbb{N}$ be a semilinear function. Then there is a finite set $\{f_1 : \mathbb{N}^k \to \mathbb{N}, \ldots, f_m : \mathbb{N}^d \to \mathbb{N}\}$ of affine partial functions, where each $\mathrm{dom} f_i$ is a linear set, and $\mathrm{dom} f_i \cap \mathrm{dom} f_j = \varnothing$ for all $i \neq j$, such that, for each $x \in \mathbb{N}^k$, if $\vec{x} \in \mathrm{dom} f_i$, then $f(\vec{x}) = f_i(\vec{x})$, and $\bigcup_{i=1}^m \mathrm{dom} f_i = \mathbb{N}^k$.
\end{lemma}

The next theorem shows that semilinear functions can be reverse-robustly computed.
\begin{theorem}
    Let $f : \mathbb{N}^k \to \mathbb{N}$ be an semilinear function. Then there is a CRC that reverse-robustly computes $f$.
\end{theorem}
We employ the construction from \cite{DBLP:journals/nc/DotyH15}. Let the set of input species be $\Sigma = \{X_1, \ldots, X_k\}$ and the output species be $Y$. By Lemma \ref{lem:decompose_semilinear}, $f$ can be decomposed into partial affine functions $f_1, \ldots, f_m$ with disjoint linear domains, and by Theorem \ref{thm:semilinear-rr-decidable} and Lemma \ref{lem:affine-diff-representation}, each of the partial affine functions and their domains can be computed reverse-robustly. By duplicating $X_{i}$, we can reverse-robustly decide the predicate $\vec{x} \in \mathrm{dom} f_j$ and compute $\hat{f}_j(\vec{x})$ in parallel for each $f_j$. Let $V_j^Y$ and $V_j^N$ be the yes and no voters, respectively, for the CRD $\mathcal{D}_j$ reverse-robustly deciding if $\vec{x} \in \mathrm{dom} f_j$ and let $\hat{Y}_j^P$ and $Y_j^C$ be the output species for the CRC $\mathcal{C}_j$ reverse-robustly computing $\hat{f}_j$, using the constructions from the previous sections. We interpret each $\hat{Y}^P_j$ and $\hat{Y}^C_j$ as ``inactive'' versions of the ``active'' output species $Y^P_j$ and $Y^C_j$. Now, we convert the function result of the applicable partial affine function to the global output by adding the following reactions for each $j \in \{1, \ldots m\}$ and for each $L_j^Y$ that votes yes and $L_j^N$ that votes no in $\mathcal{D}_j$:
\begin{align}
    L_j^Y + \hat{Y}^P_j & \rxn L_j^Y + Y^P_j + Y 
    \label{rxn:semilinear-global-Y-P} \\
    L_j^N + Y^P_j + Y & \rxn L_j^N + \hat{Y}^P_j
    \label{rxn:semilinear-global-N-P}
\end{align}
Reaction (\ref{rxn:semilinear-global-Y-P}) produces an output copy of species $Y$ and (\ref{rxn:semilinear-global-N-P}) reverses the first reaction. Both are catalyzed by the vote of the $i$-th predicate result. We also add the reactions:
\begin{align}
    L_j^Y + \hat{Y}^C_j & \rxn L_j^Y + Y^C_j
    \label{rxn:semilinear-global-Y-C} \\
    L_j^N + Y^C_j & \rxn L_j^Y + \hat{Y}^C_j
    \label{rxn:semilinear-global-N-C}
\end{align}
and
\begin{align}
    Y^P_j + Y^C_j + Y \rxn \varnothing
    \label{rxn:semilinear-global-cancel}
\end{align}
\begin{proof}
We now prove that this construction reverse-robustly decides $f$. Let $\vec{i}$ be an arbitrary initial configuration, where $\vec{i}(X_i) = x_i$, and let $\vec{c}$ such that $\vec{i} \bireach \vec{c}$. We first run each $\mathcal{D}_j$ and $\mathcal{C}_j$ in parallel until no more of the inactive output species can be produced. We then run reactions (\ref{rxn:semilinear-global-Y-P}) through (\ref{rxn:semilinear-global-cancel}) until no more reactions can occur; this resulting configuration $\vec{o}$ must be stable.

To see that $\vec{o}$ also has the correct output, since each $\mathcal{D}_j$ reverse-robustly decides $\vec{x} \in \mathrm{dom} f_j$ and the counts of their voters are not changed by reactions (\ref{rxn:semilinear-global-Y-P}) through (\ref{rxn:semilinear-global-cancel}), by Lemma \ref{lem:parallel-comp}, their individual votes are stable and correct in $\vec{o}$. Although each $\mathcal{C}_j$ reverse-robustly computes $\hat{f}_j$, since their output is consumed, we cannot apply Lemma \ref{lem:parallel-comp}. However, let $I_{\hat{f}_j}$ be the linear invariant derived for each $\mathcal{C}_j$. We see that for each $j \in \{1, \ldots, m\}$, appropriately renaming species, constant, and indices, the CRN has the linear invariant:
\begin{align*}
    I_j(\vec{c}) & = \sum_{i=1}^k n_{i,j}\vec{c}(X_i) + I_{\hat{f}_j}(\vec{c}) + d_j(\vec{c}(Y^P_j) - \vec{c}(Y^C_j)) \\
    & = \sum_{i=1}^k n_{i,j}\left(\vec{c}(X_i) + \vec{c}(X_{i,j}) + \sum_{p=1}^{c_{i,j}} p\vec{c}(C_{i,p,j}) + \sum_{i=1}^k \vec{c}(X_{i,j}')\right) \\
    & + \sum_{p=1}^{d_j-1} p\vec{c}(D^P_{p,j}) - \sum_{q=1}^{d_j-1} q\vec{c}(D^C_{q,j}) + d_j(\vec{c}(\hat{Y}^P_j) + \vec{c}(Y^P_j) - \vec{c}(\hat{Y}^C_j) - \vec{c}(Y^C_j)) - b_jd_j\vec{c}(B_j)
\end{align*}
By a similar argument to the partial affine case, it follows that:
\begin{align*}
    \vec{o}(\hat{Y}^P_j) + \vec{o}(Y^P_j) - \vec{o}(\hat{Y}^C_j) - \vec{o}(Y^C_j) & = b_j + \frac{1}{d_j}\sum_{i=1}^k n_{i,j}(x_i - c_{i,j}) = f_j(\vec{x})
\end{align*}
We further see by inspecting the reactions that the CRN has the additional linear invariant:
\begin{align*}
    I_0(\vec{c}) & = \sum_{j=1}^m \vec{c}(Y^P_j) - \vec{c}(Y)
\end{align*}
For an initial configuration $\vec{i}$ where $\vec{i}(X_i) = x_i$, the invariant has value:
\begin{align*}
    I_0(\vec{i}) = 0
\end{align*}
We argue that in $\vec{o}$, there is exactly one $j^*$ such that $L_{j^*}^Y$ is present, and that $\vec{x} \in \mathrm{dom} f_{j^*}$; this follows from each $\mathcal{D}_j$ reverse-robustly deciding their predicate, and their voters deactivating the incorrect output species. It must then be the case that there is exactly one $j^*$ (the same $j^*$) such that $Y_{j^*}^P$ and $Y_{j^*}^C$ can be present and further, that $\hat{Y}_{j^*}^P$ and $\hat{Y}_{j^*}^C$ are absent, as otherwise they could be activated by the voters of $\mathcal{D}_{j^*}$. It then follows that:
\begin{align*}
    \vec{o}(Y^P_{j^*}) - \vec{o}(Y^C_{j^*}) & = f_{j^*}(\vec{x})
\end{align*}
Additionally, the invariant $I_0$ has the value:
\begin{align*}
    I_0(\vec{o}) & = \vec{o}(Y^P_{j^*}) - \vec{o}(Y)
\end{align*}
For this to be true, $\vec{o}(Y^P_{j^*}) = \vec{o}(Y)$. Note that it is impossible for $Y^P_{j^*}$, $Y^C_{j^*}$, and $Y$ to be simultaneously present, as otherwise reaction (\ref{rxn:semilinear-global-cancel}) could occur. If $\vec{o}(Y^P_{j^*}) = \vec{o}(Y) > 0$, $\vec{o}(Y^C_{j^*}) = 0$. On the other hand, if $\vec{o}(Y^P_{j^*}) = \vec{o}(Y) = 0$, since $f_{j^*}(\vec{x}) \geq 0$, $\vec{o}(Y^C_{j^*}) = 0$. In either case, $\vec{o}(Y^C_{j^*}) = 0$. Then:
\begin{align*}
    \vec{o}(Y) = \vec{o}(Y^P_{j^*}) & = f_{j^*}(\vec{x})
\end{align*}
Since $\vec{x} \in \mathrm{dom} f_{j^*}$, $\vec{o}(Y) = f(\vec{x})$, and the CRC reverse-robustly computes $f$.
\end{proof}

%% file: conclusion.tex
\section{Conclusion}
We explored the capabilities of chemical reaction networks (CRNs) under the reverse-robust computation model, where reactions are allowed to occur in reverse while maintaining the existence of a reachable correct stable configuration.
This model aligns with thermodynamic theory assuming all species have a finite free energy, implying all chemical reactions are reversible, albeit at some possibly very slow rate.
We showed that CRNs have the same computational power under the reverse-robust computation model as they do under the stable computation model. Both semilinear predicates can be decided and semilinear functions can be computed reverse-robustly without an initial leader.
% We also introduced a new tool, invariants, for proving correctness under the reverse-robust computation model: linear (or linear modulo some $m$) combinations of the counts of the input species that are preserved by all reactions.

Some key questions remain open: first, are there a set of conditions that are necessary and/or sufficient to show that a CRN is reverse-robust? This paper was able to use existing constructions designed for the stable computation model to prove results for the reverse-robust computation model. In fact, we were unable to design a CRN that works under the former model but not the latter without the presence of a ``trap'' reaction that cannot occur when only forward reactions are allowed. We thus conjecture that the absence of such reactions is sufficient to prove reverse-robustness (and so, that a stably computing CRN can be made reverse-robust by removing such reactions without affecting its correctness). Second, our model only allows reverse-reactions to occur transiently, as we require that only forward reactions occur in the final path to reach a stable state. Is it possible to formally define some reasonable model of computation in which reverse-reactions are always allowed to occur? Note that such a model cannot require stability, as it is always possible to reverse the path of reactions that led to a state. Instead, we may require ``favorability''; although temporary reversal to incorrect states can occur, the correct answer is present ``most of the time''.

% \todoi{DD: more fuzzy open question: our model ``assumes reverse reactions eventually stop happening'' (due to requiring only forward reactions in the final path to the stable state); is there some reasonable model of CRN computation in which reverse reactions can always occur? Note that we could not hope to require stability in this case, so a new definition of computation would be required, for example, asking that the correct answer eventually be present ``most of the time'' while allowing temporarily reversal to incorrect states.}

%% file: ref.bib
@article{DBLP:journals/nc/DotyH15,
  author       = {David Doty and
                  Monir Hajiaghayi},
  title        = {Leaderless deterministic chemical reaction networks},
  journal      = {Nat. Comput.},
  volume       = {14},
  number       = {2},
  pages        = {213--223},
  year         = {2015},
  url          = {https://doi.org/10.1007/s11047-014-9435-8},
  doi          = {10.1007/S11047-014-9435-8},
  bibsource    = {dblp computer science bibliography, https://dblp.org}
}

@article{DBLP:journals/nc/ChenDS14,
  author       = {Ho{-}Lin Chen and
                  David Doty and
                  David Soloveichik},
  title        = {Deterministic function computation with chemical reaction networks},
  journal      = {Nat. Comput.},
  volume       = {13},
  number       = {4},
  pages        = {517--534},
  year         = {2014},
  url          = {https://doi.org/10.1007/s11047-013-9393-6},
  doi          = {10.1007/S11047-013-9393-6},
  bibsource    = {dblp computer science bibliography, https://dblp.org}
}

@article{DBLP:journals/dc/AngluinADFP06,
  author       = {Dana Angluin and
                  James Aspnes and
                  Zo{\"{e}} Diamadi and
                  Michael J. Fischer and
                  Ren{\'{e}} Peralta},
  title        = {Computation in networks of passively mobile finite-state sensors},
  journal      = {Distributed Comput.},
  volume       = {18},
  number       = {4},
  pages        = {235--253},
  year         = {2006},
  url          = {https://doi.org/10.1007/s00446-005-0138-3},
  doi          = {10.1007/S00446-005-0138-3},
  bibsource    = {dblp computer science bibliography, https://dblp.org}
}

@inproceedings{DBLP:conf/wdag/DotyH24,
  author       = {David Doty and
                  Ben Heckmann},
  editor       = {Dan Alistarh},
  title        = {The Computational Power of Discrete Chemical Reaction Networks with
                  Bounded Executions},
  booktitle    = {38th International Symposium on Distributed Computing, {DISC} 2024,
                  Madrid, Spain, October 28 - November 1, 2024},
  series       = {LIPIcs},
  pages        = {20:1--20:15},
  publisher    = {Schloss Dagstuhl - Leibniz-Zentrum f{\"{u}}r Informatik},
  year         = {2024},
  url          = {https://doi.org/10.4230/LIPIcs.DISC.2024.20},
  doi          = {10.4230/LIPICS.DISC.2024.20},
  bibsource    = {dblp computer science bibliography, https://dblp.org}
}
